\documentclass[pra,twocolumn,showpacs,preprintnumbers,amsmath,amssymb,superscriptaddress]{revtex4}

\usepackage{bm}
\usepackage{graphicx}
\usepackage{amsbsy}
\usepackage{amsmath}
\usepackage{amsfonts}
\usepackage{amsthm}

\begin{document}

\theoremstyle{plain}
\newtheorem{theorem}{Theorem}
\newtheorem{lemma}[theorem]{Lemma}
\newtheorem{corollary}[theorem]{Corollary}
\newtheorem{conjecture}[theorem]{Conjecture}
\newtheorem{proposition}[theorem]{Proposition}

\theoremstyle{definition}
\newtheorem{definition}{Definition}

\theoremstyle{remark}
\newtheorem*{remark}{Remark}
\newtheorem*{example}{Example}

\title{Necessary conditions on entanglement catalysts}   
\author{Yuval Rishu Sanders}\email{yrsanders@math.ucalgary.ca}
\author{Gilad Gour}
\affiliation{Institute for Quantum Information Science and 
Department of Mathematics and Statistics,
University of Calgary, 2500 University Drive NW,
Calgary, Alberta, Canada T2N 1N4} 

\date{\today}

\begin{abstract}
Given a pure state transformation $\psi\mapsto\phi$ restricted to entanglement-assisted local operations with classical communication, we determine a lower bound for the dimension of a catalyst allowing that transformation. Our bound is stated in terms of the generalised concurrence monotones (the usual concurrence of two qubits is one such monotone). We further provide tools for deriving further conditions upon catalysts of pure state transformations.
\end{abstract}

\pacs{03.67.Mn, 03.67.Hk, 03.65.Ud}

\maketitle

\section{Introduction}
\label{sectINT}

Originally thought a paradox, the inseparable quality of many-body quantum systems known today as \emph{entanglement} has been shown to be fundamental to understanding the emerging field of Quantum Information. A popular example is the so-called \emph{quantum teleportation} process, in which a precise quantum state may be transmitted between labs using only \emph{local operations} assisted by \emph{classical communication} (usually abbreviated as LOCC). To perform this task, the two labs must `consume' a shared entangled state: they must make an entangled system `less' entangled in a precise sense. This precise notion of comparison of entangled states is accomplished by \emph{Nielsen's Theorem}~\cite{Nielsen:1999}.

Nielsen's Theorem essentially states that there is a partial order $\prec$ on the set of states in some bipartite Hilbert space $\mathcal{H}^A \otimes\mathcal{H}^B = \mathcal{H}^{AB}$ which is invariant under local unitary operations (i.e. unitary operators of the form $U^A \otimes U^B \in \mathcal{B}(\mathcal{H}^A)\otimes\mathcal{B}(\mathcal{H}^B)$) and has the property that, for any $T\in \mathcal{B}(\mathcal{H}^{AB})$, we have $\psi\prec T(\psi)$ if and only if $T$ is LOCC. These conditions on the partial order $\prec$ (known as \emph{majorization}) are valuable in studying entangled states in general because entanglement cannot increase under LOCC restrictions. It is thus possible to say that $\phi$ is less entangled than $\psi$ if $\psi\prec\phi$. The converse is not true in general. 

If two labs share some (pure) entangled state $\psi$, there generally exist states $\phi$ which cannot be constructed from $\psi$ using only LOCC. If those two labs share an additional entangled state $\chi\in\mathcal{H}^{A'B'}$ (the $'$ represents the fact that this is a different Hilbert space from $\mathcal{H}^{AB}$ in which $\psi$ is contained), it is possible to consume some of the entanglement of $\chi$ to allow the labs to create some desired $\phi$ they could not create before~\cite{JonathanPlenio:1999}. The problem is that there exist $\phi$ so that, while $\psi\mapsto\phi$ may be impossible under LOCC, we have $\psi\otimes\chi\prec\phi\otimes\chi$ and thus an LOCC transformation by Nielsen's Theorem. We can interpret this as saying that $\psi\otimes\chi\in\mathcal{H}^{AB}\otimes\mathcal{H}^{A'B'}$ can be mapped to $\phi\otimes\chi\in\mathcal{H}^{AB}\otimes\mathcal{H}^{A'B'}$, or that we can perform $\psi\mapsto\phi$ without consuming \emph{any} extra entanglement.

The problem of \emph{entanglement catalysis} is to quantify the resource of access (without consumption) of ancillary entangled states. It has been recently shown~\cite{Turgut:2007,Klimesh:2007} that a pure state transformation $\psi\mapsto\phi$ is possible under \emph{entanglement-assisted} LOCC, or eLOCC, if and only if $S_\alpha (\sigma(\psi)) \geq S_\alpha (\sigma(\phi))$ for each $\alpha>0$. Here, $\sigma(\zeta)$ denotes the \emph{Schmidt coefficients} of the pure state $\zeta$ and $S_\alpha (x)$ denotes the \emph{R\'{e}nyi entropy} of order $\alpha$ of some discrete probability distribution $x$. Explicitly,
\begin{equation} \label{renyi} S_\alpha (x) = \frac{1}{1-\alpha} \log \left( \sum_i x_i^\alpha \right) \end{equation}
where $\lim_{\alpha\rightarrow 1} S_\alpha (x) = -\sum_{i} x_i \log x_i$ is the usual \emph{Shannon entropy} of $x$.

It is thus possible, given states $\psi, \phi$, to check whether there exists an eLOCC operation mapping one state to the other. This does not fully characterise the transformation, however, because this test gives no information about the catalyst itself. Indeed, for any ancillary state $\chi$, we have $S_\alpha (\psi\otimes\chi) = S_\alpha (\psi) + S_\alpha (\chi)$. While it is possible in principle to test $\psi\otimes\chi$ and $\phi\otimes\chi$ (for arbitrary $\chi$) against Nielsen's Theorem, such a direct test is usually impractical. What is needed is some general method for obtaining information about possible catalysts given only $\psi$ and $\phi$.

In this paper we present, for the first time, conditions on possible catalysts of some eLOCC transformation. In particular, we provide a lower bound on the dimension of a possible catalyst state. This lower bound is stated and proved by extensive use of the \emph{generalised concurrence monotones}~\cite{Gour:2005} (or \emph{concurrences} for short). The concurrences are a set of $N-1$ measures of entanglement (in the sense of Vidal~\cite{Vidal:2000}) of some $N$-dimensional state, and expand upon the more well-known concurrence of two-dimensional states~\cite{HillWootters:1997}. Section~\ref{sectCONC} introduces the concurrences and the general procedure for providing the lower bound. In section~\ref{sectMAIN} we state and prove the fundamental proposition of this paper; namely, the lower bound on the dimension of a possible catalyst of an eLOCC transformation. We also discuss how our techniques can be used to find further conditions on such catalysts. Section~\ref{sectEND} provides an example of an eLOCC transformation where our methods may assist us in finding a catalyst.

\section{Concurrences and concurrence factorisation}
\label{sectCONC}

\begin{definition}[Generalised concurrence monotones~\cite{Gour:2005}]
The $k^{th}$ \emph{concurrence} of a pure bipartite state $\zeta$ (of dimension $n\geq k$) is defined as
\begin{equation}\label{concs} C_k (\zeta) \equiv \left( \frac{ e_k (\sigma(\zeta)) }{ e_k \left( \iota_n  \right) } \right)^{1/k} \end{equation}
where $\iota_n = (1/n,\ldots,1/n) \in \mathbb{R}^n$, $\sigma(\zeta)$ denotes the Schmidt coefficients of the state $\zeta$, and the $k^{th}$ \emph{elementary symmetric polynomial} $e_k (x)$ of $n$ variables $x=(x_1,\ldots,x_n)$ is defined as
\begin{equation}\label{eks} e_k (x) \equiv \sum_{i_1 < \cdots < i_k} x_{i_1} \cdots x_{i_k} \end{equation}
with $e_0 \equiv 1$. If $\zeta$ is mixed, we define \[C_k (\zeta) = \min_{\{p_i, \zeta_i\}} \sum_i p_i C_k(\zeta_i)\] where the minimization is over all pure state ensembles $\{p_i, \zeta_i\}$ realizing $\zeta$.
\end{definition}

Given some LOCC transformation $\psi\mapsto\phi$, it is a necessary condition that each concurrence is monotonic, i.e. $C_k (\psi) \geq C_k (\phi)$ for $k=2\ldots n$ if the dimension of $\psi$ is $n$. Thus, if $\psi\mapsto\phi$ by eLOCC, we must have $C_k (\psi\otimes\chi) \geq C_k (\phi\otimes\chi)$ for any catalyst $\chi$ so that $\psi\otimes\chi\mapsto\phi\otimes\chi$ by LOCC. The goal of this section is to evaluate this inequality in terms of $C_h (\psi), C_i (\phi), C_j (\chi)$ for various indices $h,i,j$.

The elementary symmetric polynomials are, as their name suggests, the most natural choice of a basis for the ring of symmetric polynomials (polynomials invariant under permutation of the variables). Another common basis is the set of \emph{power sum symmetric polynomials} $p_l (x) \equiv \sum_i x_i^l$ which have the enviable property that $p_l (x\otimes y) = p_l (x) p_l (y)$. To exploit this property, we make use of \emph{Newton's Identities}~\cite{MathWorld:NewtId}:
\begin{equation}\label{newtid} k e_k (x) = \sum_{l=1}^{k} (-1)^{l-1} e_{k-l} (x) p_{l} (x) \end{equation}
It is a straightforward task to recursively use this expression to write $e_k$ as a function of $p_l$ and vice-versa.

\begin{equation}\label{ekaspl}
\begin{array}{ccl}
e_1 & = & p_1 \\
e_2 & = & \frac{1}{2} \left(p_1^2-p_2\right) \\
e_3 & = & \frac{1}{6} \left(p_1^3-3 p_1 p_2+2 p_3\right) \\
e_4 & = & \frac{1}{24} \left(p_1^4-6 p_1^2 p_2+3 p_2^2+8 p_1 p_3-6 p_4\right) \\
 & \vdots & \\
\end{array}
\end{equation}

\begin{equation}\label{plasek}
\begin{array}{ccl}
p_1 & = & e_1 \\
p_2 & = & e_1^2-2e_2 \\
p_3 & = & e_1^3-3e_1e_2+3e_3 \\
p_4 & = & e_1^4-4 e_1^2 e_2+2 e_2^2+4 e_1 e_3-4 e_4 \\
 & \vdots & \\
\end{array}
\end{equation}

We can therefore take $e_k (x\otimes y)$ to be a polynomial of various $p_l (x\otimes y) = p_l (x) p_l (y)$. We may then write each $p_l (x)$ and $p_{l'} (y)$ as polynomials of $e_j (x)$ and $e_{j'} (y)$. Explicitly, we have the following factorisations: 
\begin{equation}\label{factid}
\begin{array}{rcl}
e_1 (x\otimes y) & = & e_1 (x) e_1 (y) \\
e_2 (x\otimes y) & = & e_1 (x) ^2 e_2 (y) + e_2 (x) e_1 (y) ^2 \\
 & & - 2 e_2 (x) e_2 (y) \\
e_3 (x\otimes y) & = & e_3 (x) e_1 (y) ^3 + e_1 (x) ^3 e_3 (y) \\
 & & + e_1 (x) e_2 (x) e_1 (y) e_2(y)\\
 & & - 2 e_1 (x) e_2 (x) e_3 (y) \\
 & & - 2 e_3 (x) e_1(y) e_2(y) + 3 e_3 (x) e_3 (y) \\
 & \vdots & \\
 e_{d_1 d_2-1} (x\otimes y) & = & e_{d_1} (x) ^{d_2-1} e_{d_2} (y) ^{d_1-1} e_{d_1-1} (x)  e_{d_2-1} (y) \\
 e_{d_1 d_2} (x\otimes y) & = & e_{d_1} (x) ^{d_2} e_{d_2} (y) ^{d_1} \\
\end{array}
\end{equation}
where $d_1$ and $d_2$ are the number of non-zero components of $x$ and $y$, respectively.
From these equations and the fact that $e_1(x)=e_1(y)=1$, since we are consiering only normalised states, we may immediately deduce that two common entanglement measures, the \emph{I-concurrence} (given by $C_2$) and the \emph{G-concurrence} (given by $C_d$ of a $d$-dimensional state) are monotones under catalysis.
It is also from the equations above that our main result, Proposition~\ref{mainres}, shall follow. 

Before presenting our main result, it is important to note that the expressions for $e_k(x\otimes y)$ are much simpler for $k$ close to $1$ or $d_1d_2$. This behavior can be understood from the following lemma.

\begin{lemma}
\label{thelemma}
Suppose $(x_i) = x \in \mathbb{R}^d$ such that $x_i \neq 0\ (\forall i)$. Define $1/x \equiv \left(\frac{1}{x_i}\right)$. Then $e_k (1/x) = e_{d-k}(x)/e_d(x)$.
\end{lemma}

\begin{proof}
\[ e_{d-k} (x) = \sum_{i_1<\cdots<i_k} \frac{x_1 \cdots x_d}{x_{i_1} \cdots x_{i_k}} = e_d(x)e_k(1/x) \]
\end{proof}

This simple lemma will also be useful in proving Proposition~\ref{mainres}.

\section{Bounding the dimension of a catalyst}
\label{sectMAIN}

In this section we use Newton's identities and in particular Eq.~(\ref{factid}) to provide
conditions on the catalysis. We start with a lower bound on the dimension.

\begin{proposition}
\label{mainres}
Suppose $\psi\otimes\chi \mapsto \phi\otimes\chi$ by LOCC and that $\psi,\phi,\chi$ are pure. If $\psi$ and $\phi$ both have precisely $d$ non-zero Schmidt coefficients and $\chi$ has $b$ non-zero Schmidt coefficients, then
\[ b \geq 1 + \left(\frac{d-1}{d}\right) \frac{ \log \left(C_{d-1} (\phi)\right) - \log \left(C_{d-1} (\psi)\right) }{ \log \left(C_d (\psi)\right) - \log \left(C_d (\phi)\right) } \]
\end{proposition}

\begin{remark}
Notice that the bound given is nontrivial only if the ratio is positive. We can be assured that $\log \left(C_d (\psi)\right) - \log \left(C_d (\phi)\right)$ is positive because of the final equality of equation~(\ref{factid}). Indeed, it says that $C_{db} (\psi\otimes\chi) = C_d (\psi) C_b (\chi)$, so $C_{db} (\psi\otimes\chi) \geq C_{db} (\phi\otimes\chi) \Rightarrow C_{db} (\psi) \geq C_{db} (\phi)$ (it is also immediate from equation~(\ref{factid}) that $C_2$ is monotonic). Therefore, we require that $C_{d-1} (\psi) < C_{d-1} (\phi)$ for this bound to be nontrivial. Such examples can be found, although verifying their R\'{e}nyi entropy monotonicity (and thus the existence of a catalyst) is necessary on a case-by-case basis.
\end{remark}

\begin{proof}
From Lemma~\ref{thelemma} we have \[ e_{d b - 1} (\sigma(\psi\otimes\chi)) = e_1(1/\sigma(\psi\otimes\chi)) e_{db} (\sigma(\psi\otimes\chi)) \]
which, by equation~(\ref{factid}), may be rewritten as
\[ e_1 (1/\sigma(\psi)) e_1(1/\sigma(\chi)) e_{d} (\sigma(\psi)) ^{b} e_{b} (\sigma(\chi)) ^{d} .\]
Lemma~\ref{thelemma} further implies that
\[e_1 (1/\sigma(\psi))=e_{d-1} (\sigma(\psi))/e_{d} (\sigma(\psi)) .\]
Since $C_{db-1} (\psi\otimes\chi) \geq C_{db-1} (\phi\otimes\chi)$, a direct computation reveals that
\[ C_{d-1}^{d-1} (\psi) C_d^{d(b-1)} (\psi) \geq C_{d-1}^{d-1} (\phi) C_d^{d(b-1)} (\phi) \]
and taking the logarithm of this inequality proves the proposition.
\end{proof}

Our lower bound is thus found by asking that the second-to-last concurrence of $\psi\otimes\chi$ be greater than that of $\phi\otimes\chi$ and that this inequality be reversed for $\psi$ and $\phi$. It is therefore pertinent to ask what happens in the case of other violations of concurrence monotonicity. Indeed, it is possible to derive further conditions based on this assumption, but these conditions seem to be more difficult to analyse. 
Two examples of such conditions are presented as Propositions~\ref{res3} and~\ref{res2}. 

\begin{proposition}
\label{res3}
Suppose $\psi\otimes\chi \mapsto \phi\otimes\chi$ by LOCC and that $\psi,\phi,\chi$ are pure. Let
\begin{align}
& r(\chi) \equiv\frac{e_2(\sigma(\chi))-2e_3(\sigma(\chi))}{1-2e_2(\sigma(\chi))+3e_3(\sigma(\chi))}\nonumber\\
& a(\psi,\phi) \equiv e_2(\sigma(\psi))-e_2(\sigma(\phi))\nonumber\\
& b(\psi,\phi) \equiv e_3(\sigma(\psi))-e_3(\sigma(\phi))\nonumber\;.
\end{align}
Then,
\begin{equation}\label{tt}
r(\chi) \geq -\frac{b(\psi,\phi)}{a(\psi,\phi)}
\end{equation}
\end{proposition}

\begin{remark}
From basic properties of the elementary symetric functions it follows that
the function $r(\chi)$ is always non-negative. Also, since $e_2$ is a monotone under eLOCC
the condition $\psi\otimes\chi \mapsto \phi\otimes\chi$ by LOCC
implies that $a(\psi,\phi)\geq 0$. On the other hand, the function $e_3$ is not a monotone under eLOCC and
therefore $b(\psi,\phi)$ can be both positive or negative. The bound is non-trivial when $b$ is negative.
\end{remark}

\begin{proof}
The proof follows from the fact the third concurrence (like all the other concurrences) is an entanglement monotone.
This implies that
$$
e_{3} (\sigma(\psi\otimes\chi))\geq e_{3} (\sigma(\phi\otimes\chi))\;.
$$
From the equation above and the expression for $e_3$ in Eq.(\ref{factid}) we find 
\begin{align}
b(\psi,\phi) & +a(\psi,\phi)e_2(\sigma(\chi))-2a(\psi,\phi)e_3(\sigma(\chi))\nonumber\\
& -2b(\psi,\phi)e_2(\sigma(\chi))+3a(\psi,\phi)e_3(\sigma(\chi))\geq 0\;.\nonumber
\end{align}
Rearrangement of the terms in the equation above leads to Eq.(\ref{tt}).
\end{proof}

From the monotonicity 
$$
e_{k} (\sigma(\psi\otimes\chi))\geq e_{k} (\sigma(\phi\otimes\chi))\;.
$$
one can obtain further conditions on the catalysis. However, these conditions become more and more complicated as
$k$ becomes closer to $db/2$, where $d$ is the dimension of $\psi$ and $b$ the dimension of the catalyst $\chi$.
For instance, for $k=db-2$ we get the following bound: 

\begin{proposition}
\label{res2}
Suppose $\psi\otimes\chi \mapsto \phi\otimes\chi$ by LOCC and that $\psi,\phi,\chi$ are pure. If $\psi$ and $\phi$ both have precisely $d$ non-zero Schmidt coefficients and $\chi$ has $b$ non-zero Schmidt coefficients, then
\[ 
\left( \frac{C_{b-1} (\chi)}{C_{b-2} (\chi)} \right)^{d-2} \geq \frac{b-1}{b}
\left(\frac{C_{d-2} (\phi)}{C_{d-2} (\psi)}\right)^{d-1} \left(\frac{C_d (\psi)}{C_d (\phi)}\right)^d \Lambda \]
with
\[ \Lambda \equiv \frac{1}{bd} \left( \rho_{2,d} (\psi) - \rho_{2,d} (\phi) \right) - \frac{d}{d-1} \left( \rho_{1,d} (\psi)^2 - \rho_{1,d} (\phi)^2 \right) \]
and
\[ \rho_{k,d} (\xi) \equiv \frac{C_{d-k} (\xi)^{d-k}}{C_{d} (\xi)^d} \]
\end{proposition}

Propositions~\ref{res3} and \ref{res2} assure us that it is indeed possible to bound the entanglement of a possible catalyst, but the bound depends on the number $b$ of non-zero Schmidt coefficients of $\chi$. This limits the usefulness of the proposition and, in all likelihood, of some other expressions derivable from equation~(\ref{factid}). It is still remarkable, however, that the concurrences of the catalyst can be approached using our methods.

\begin{example}\label{mainex}
Consider two states $\psi$ and $\phi$ with Schmidt vectors
\[ \sigma(\psi) = \left( \frac{19}{351}, \frac{1}{13}, \frac{64}{351}, \frac{71}{351}, \frac{3}{13}, \frac{89}{351} \right) \]
\[ \sigma(\phi) = \left( \frac{9}{196}, \frac{25}{196}, \frac{13}{98}, \frac{5}{28}, \frac{3}{14}, \frac{59}{196} \right) \]

We immediately see that there is no LOCC transformation between these two states since $19/351>9/196$ but $19/351+1/13<9/196+25/196$ and yet it is possible to verify numerically that there exists $\chi$ with $\psi\otimes\chi \mapsto \phi\otimes\chi$ by finding the roots of $f(\alpha):= S_\alpha (\sigma(\psi)) - S_\alpha (\sigma(\phi))$~\cite{Turgut:2007,Klimesh:2007} (the only root is $\alpha\rightarrow0$). Proposition~\ref{mainres} informs us that such a catalyst must have dimension greater than about $2.7$, so in particular we have found that no catalysts of dimension $2$ exist for the eLOCC transformation $\psi\mapsto\phi$. Supposing that a catalyst $\chi$ of dimension three exists, Proposition~\ref{res2} asserts that $C_2 (\chi)\geq 0.436$.
\end{example}

\section{Conclusions}
\label{sectEND}

The main problem we have considered in this paper is the following: supposing we have an eLOCC transformation $\psi\mapsto\phi$, what catalysts $\chi$ give us an LOCC transformation $\psi\otimes\chi\mapsto\phi\otimes\chi$? We provide a partial answer based on analysis of the generalised concurrence monotones. Our analysis is based on the ability to express the concurrences of a tensor product state ($\psi\otimes\chi$) in terms of the concurrences of its component states ($\psi$ and $\chi$) and provides us with, in particular, a lower bound on the dimension of a possible catalyst.

Further conditions are derivable from the fact that some concurrences are not monotones under eLOCC despite all being entanglement monotones. In fact, our conditions require this non-monotonicity in order to be nontrivial. In this sense, our solution is only partial: there is still no general method for finding a catalyst for a given eLOCC transformation. We can only give useful information in the case that there is non-monotonicity in one of the concurrences of $\psi$ and $\phi$. Nevertheless, such cases exist, and we have given an example where our methods are useful.

To proceed from this work to a more general theory of eLOCC transformations, it would be necessary to fully characterise the behaviour of the concurrences (and perhaps other entanglement monotones) under eLOCC maps. It is important to understand such transformations because they provide extra conversion power for quantum states with almost no increase in resource cost (one catalyst state is reusable if one wishes to perform many of the same conversions). Such conversion power can be extremely valuable, since most quantum informational tasks require a specific entangled state. In order to provide that specific state, the process of interconverting entangled states must be understood and optimized as best as possible. Our techniques are a step in that direction.

\emph{Acknowledgments:---}
The authors acknowledge support from the National Science and Engineering Research Council.


\begin{references}
\bibitem{Nielsen:1999} M. A. Nielsen, ``Conditions for a class of entanglement transformations'', Phys.~Rev.~Lett. \textbf{83} (2) 436-439 (1999)
\bibitem{JonathanPlenio:1999} D. Jonathan and M. B. Plenio, ``Entanglement-assisted local manipulation of pure quantum states'', Phys.~Rev.~Lett. \textbf{83} 3566-3569 (1999)
\bibitem{Turgut:2007} S. Turgut, ``Necessary and Sufficient Conditions for the Trumping Relation'', J.~Phys.~A~\textbf{40} 12185-12212 (2007)
\bibitem{Klimesh:2007} M. Klimesh, ``Inequalities that Collective Completely Characterize the Catalytic Majorization Relation'', arXiv:0709.3680v1 (2007)
\bibitem{Gour:2005} G. Gour, ``Family of Concurrence Monotones and its Applications'', Phys.~Rev.~A \textbf{71} 012318 (2005)
\bibitem{Vidal:2000} G. Vidal, ``Entanglement Monotones'', J.~Mod.~Opt. \textbf{47} 355 (2000)
\bibitem{HillWootters:1997} S. Hill and W. K. Wootters, ``Entanglement of a Pair of Quantum Bits'',  Phys.~Rev.~Lett. \textbf{78} 5022-5025 (1997)
\bibitem{MathWorld:NewtId} Weisstein, Eric W. ``Newton-Girard Formulas.'' From MathWorld -- A Wolfram Web Resource. http://mathworld.wolfram.com/Newton-GirardFormulas.html
\end{references}
\end{document}